\documentclass[conference]{IEEEtran}
\IEEEoverridecommandlockouts
\usepackage{cite}
\usepackage{amsthm,amsmath,amssymb,amsfonts}
\usepackage{braket}
\usepackage{enumerate}
\usepackage{algorithmic}
\usepackage{graphicx}
\usepackage{textcomp}
\usepackage{xcolor}
\usepackage{url}

\newcommand{\XX}{\mathcal{X}}
\newcommand{\YY}{\mathcal{Y}}
\newcommand{\LL}{\mathcal{L}}
\newcommand{\WW}{\mathcal{W}}
\newcommand{\DD}{\mathcal{D}}

\newcommand{\A}{\mathcal{A}}
\newcommand{\B}{\mathcal{B}}
\newcommand{\C}{\mathcal{C}}
\newcommand{\HH}{\mathcal{H}}

\newcommand{\KK}{\mathcal{K}}
\newcommand{\E}{\mathbb{E}}
\newcommand{\D}{\mathbb{D}}
\newcommand{\V}{\mathbb{V}}
\newcommand{\R}{\mathbb{R}}
\newcommand{\gG}{\mathfrak{G}}
\newcommand{\Ss}{\mathcal{S}}
\newcommand{\ns}{\mathit{ns}}
\newcommand{\co}{\mathit{co}}
\newcommand{\opt}{\mathrm{opt}}
\newcommand{\Ber}{\mathrm{Ber}}
\newcommand{\CHSH}{\mathit{CHSH}}

\DeclareMathOperator*{\argmax}{arg\,max}

\newcommand{\val}{\text{val}}

\newcommand{\tr}{\text{tr}}

\newtheorem{theorem}{Theorem}
\newtheorem{lemma}[theorem]{Lemma}
\newtheorem{definition}[theorem]{Definition}
\newtheorem{proposition}[theorem]{Proposition}
\newtheorem{corollary}[theorem]{Corollary}

\def\BibTeX{{\rm B\kern-.05em{\sc i\kern-.025em b}\kern-.08em
    T\kern-.1667em\lower.7ex\hbox{E}\kern-.125emX}}
\begin{document}

\title{Beware of Greeks bearing entanglement?  Quantum covert channels, 
information flow and non-local games \thanks{The author is supported by FNR under
grant INTER FNRS/15/11106658/SeVoTe.}
}

\author{David Mestel \\ University of Luxembourg \\ david.mestel@uni.lu}

\maketitle

\begin{abstract}
Can quantum entanglement increase the capacity of (classical) covert channels?
To one familiar with Holevo's Theorem it is tempting to think that the answer is
obviously no. However, in this work we show: quantum entanglement can in fact
increase the capacity of a classical covert channel, in the presence of an
active adversary; on the other hand, a zero-capacity channel is not improved by
entanglement, so entanglement cannot create `purely quantum' covert channels;
the problem of determining the capacity of a given channel in the presence of
entanglement is undecidable; but there is an algorithm to bound the entangled
capacity of a channel from above, adapted from the semi-definite hierarchy from
the theory of non-local games, whose close connection to channel capacity is at
the core of all of our results.
\end{abstract}


\section{Introduction}\label{sec:intro}
Suppose that you are processing sensitive data using a computer. How do you know
that your computer was not given to you in a state of quantum entanglement with
an eavesdropping adversary? This is a situation that (unlike the presence of
Greek soldiers) cannot be detected by any local experiment.  It does not require 
the victim to be using any kind of quantum technology.\footnote{Although we must 
acknowledge that it would require the adversary to have technical capabilities 
beyond those publicly known, since holding systems in superposition is 
currently a sensitive, fragile and usually short-lived affair.}

Fortunately, the presence of entanglement does not of itself jeopardise the 
privacy of one's data.  This is due to the `non-signalling' property of 
entanglement: although the adversary is able to obtain non-classical correlations 
with the victim's measurement outcomes, this does not allow him to deduce anything
about what those measurements were (otherwise distant entanglement would enable
faster-than-light communication).  But what if the adversary \emph{also} has 
access to some legitimate interaction with the victim, such as use of a shared 
resource?  Is it possible for entanglement to create a covert channel where none 
would otherwise exist, or to increase the power of an existing channel?  The 
purpose of the present work is to address this question.

Holevo's Theorem~\cite{holevo73theorem} states (in relevant
part) that entanglement cannot increase the classical Shannon capacity of a purely 
classical discrete memoryless channel.  It is therefore tempting to assume that 
this means the answer to the above question is `no'; however, as we shall see, 
in fact the picture is rather more complex.

An abstracted representation of a system which may or may not give rise to a
covert channel is shown in Figure \ref{fig:flowfig}. A victim, Alice, interacts
with some system $\C$, to which access is also given to an eavesdropper, Bob. Bob may
only receive messages from the system (a passive adversary), or he may also be
able to send messages (an active adversary). We say that a covert channel exists
if Bob is able to learn something about Alice's actions from his observations;
this is set out formally in the classic paper of Goguen and
Meseguer~\cite{goguen1982security}. Note that we make the assumption that Alice
is actively trying to convey information to Bob. This may be because she
(perhaps a malicious process) is trying to exfiltrate data across what should be
an information flow barrier. Alternatively she may be an innocent victim (in
this situation the covert channel is often called a `side-channel'), but if her
behaviour is not specified then a conservative analysis must assume that she
could behave as if trying to exfiltrate data.

It may be, however, that the question just of whether Bob can learn
\emph{anything} is too crude, and we may be interested in \emph{how much}
information can reach Bob from Alice; this is the subject of the field of
`Quantitative Information Flow' (QIF). The original approach~\cite{gray1992toward}
was to compute the Shannon mutual information between Alice's actions and Bob's
observations, but it was pointed out by Smith~\cite{smith2009foundations} that
this is usually inappropriate. This has given rise to extensive study of various
possible measures of information flow; see the recent
book~\cite{alvim2020information}. However, in this work we will mainly (with the
exception of Section \ref{sec:sdp}) be agnostic as to the choice of measure,
subject to mild reasonableness conditions.

The goal of QIF is essentially to analyse Figure \ref{fig:flowfig} in 
quantitative fashion.  The goal of this paper is to extend this analysis to the 
situation where Alice and Bob may share entanglement.  We define information 
flow in this setting and then address some fundamental questions.  Can 
entanglement make any difference?  Can we tell how much?  Can entanglement 
introduce covert channels where none existed before?

\subsection*{Overview}

The structure of this paper is as follows. In Section \ref{sec:prelim} we set
out basic concepts and definitions, for both classical information flow and
quantum entanglement. In Section \ref{sec:entangledcap} we define an entangled
version of information flow. We then introduce the reader to `non-local games',
an important concept from Quantum Information Theory that provides the technical
machinery for many of our results, and show by a simple reduction from a certain
game (the `CHSH game') that it is possible for entangled capacity to exceed
classical capacity (Theorem \ref{thm:quantad}). In Section
\ref{sec:purelyquant}, on the other hand, we show that if a channel has zero
classical capacity then it also has zero entangled capacity (Corollary
\ref{cor:purelyquant}), and so it is not possible for entanglement to introduce
`purely quantum' covert channels. In Section \ref{sec:noncomp}, we consider the
problem of computing the entangled capacity of a given channel, and show using
the very recent breakthrough result $\text{MIP}^*=\text{RE}$~\cite{ji2020mipre}
that the problem of computing this capacity, even to within a constant factor
approximation, is undecidable (Theorem \ref{thm:capnoncomp}). More positively,
in Section \ref{sec:sdp} we show that the Semi-Definite Programming (SDP)
methods~\cite{navascues2008convergent} for bounding the value of non-local games
can be adapted to give upper bounds for entangled channel capacity. Finally in
Section \ref{sec:conclusion} we reflect on the connection between covert channel
capacity and non-local games, and consider future directions for quantum QIF.

\begin{figure}[t]
\centering
\includegraphics[width=0.4\textwidth]{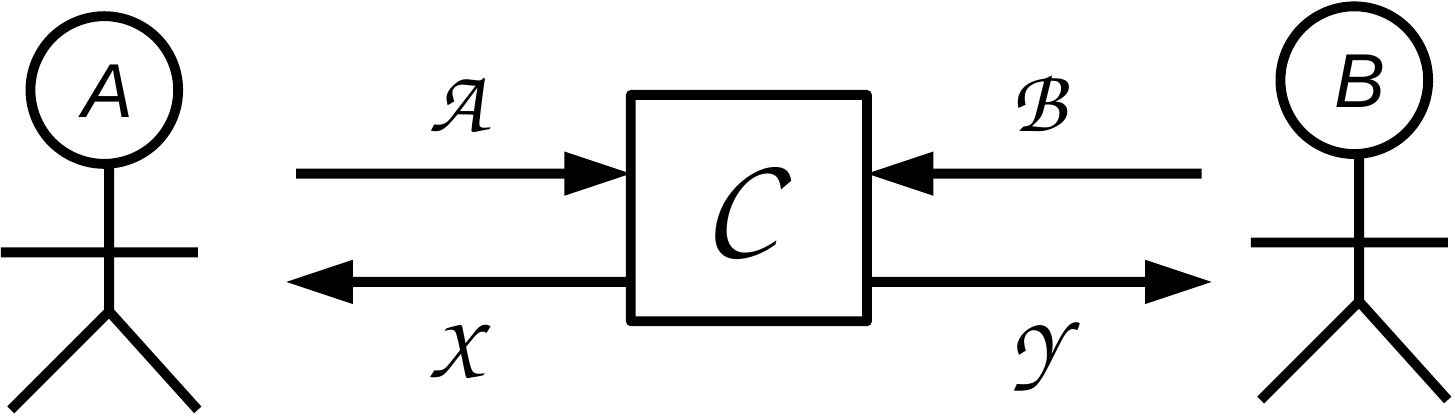}
\caption{An abstracted system}\label{fig:flowfig}
\end{figure}

\subsection*{Related work}

As far as we are aware, the only work which has attempted to extend QIF into the 
quantum realm is the paper of Am{\'e}rico and Malacaria~\cite{americo2020qqif}.
This studies a rather different setting, in which Alice sends to Bob a quantum 
state $\rho^x$ which is a (specified) function of the secret value $x\in\XX$; 
Bob is then allowed to apply a single measurement of his choice from a fixed 
set of allowed measurements, and the question is how much Bob can learn about 
the secret $x$ according to various measures of information flow.  This is of 
course only relevant to a network in which quantum states can be passed around.

The question of communication channels and their capacity is of course central
to information theory, and quantum information theory is a huge topic in modern
physics (see for instance~\cite{wilde2013quantum}). However, perhaps
surprisingly the present setting of the classical capacity of a classical fully
interactive multi-round channel assisted by entanglement has not as far as we
can tell been previously studied (see also the more recent survey
\cite{gyongyosi2018survey}). This may be because physicists are generally more
interested in quantum channels (which allow quantum states to be sent and
received), or in the effect of entanglement on the quantum capacity of classical
channels (which surprisingly can be positive due to the technique of `quantum
teleportation'). Additionally, the idea of a fully interactive channel may not
seem particularly `physical', since it is fairly far from the classic setting of
a noisy communication medium; on the other hand such a situation is common in 
the setting of covert channels or side-channels arising from use of a shared 
resource or interaction with a common system.

\section{Preliminaries}\label{sec:prelim}

\subsection{Classical information flow}\label{sec:classflow}

Although many different models (at varying levels of abstraction) have been used
in other works to represent the behaviour of the system, for this paper we will 
adopt a simple abstract model, a finite-round version of the model from the 
author's prior work~\cite{mestel2019quantifying} (and a multi-round version of 
the model used in~\cite{alvim2017information}).  We assume that Alice and 
Bob interact with the system over $k$ rounds, at each round sending a message 
drawn from finite sets $\A$ and $\B$ respectively, and receiving in return 
messages from finite sets $\XX$ and $\YY$ respectively.  The behaviour of the 
system is then specified just by functions determining the probability 
distribution on output messages, based on the actions that have occurred up to 
that point:

\begin{definition}\label{def:chan}
An \emph{$n$-round abstract interactive channel} ($n$-IC) is given by finite sets
$\A,\B,\XX,\YY$ and an $n$-tuple $(f_1,\ldots,f_n)$ of functions
\[f_i:(\A\times \B \times \XX \times \YY)^{i-1}\times (\A\times \B) \rightarrow \D(\XX\times \YY).\]
\end{definition}

Note that $\D(\XX\times \YY)$ denotes the space of probability distributions on
the set $\XX\times \YY$. Note also that no generality is lost by using the same
finite sets for each round of interaction: to represent a system using sets
$\A_i,\B_i,\XX_i,\YY_i$ at round $i$, take $\A=\sqcup_i \A_i$ (the disjoint
union of the $\A_i$), and similarly for $\B,\XX$ and $\YY$. Choose arbitrary
$a_i\in \A_i$ and $b_i\in \B_i$, and set the images of the $f_i$ to be supported
only on $\XX_i\times \YY_i$ and treat inputs at round $i$ which are not in
$\A_i$ (respectively $\B_i$) as though they were $a_i$ (respectively $b_i$).

A simple example of such a system is a fair resource scheduler, which receives 
requests from Alice and Bob and (if both ask to use the resource) allocates the
resource to whichever has received it fewer times in the past (breaking ties 
randomly).  This has $\A=\B=\XX=\YY=\{0,1\}$, and $f_i(t,(1,0))=(1,0), 
f_i(t,(0,1)) = (0,1), f_i(t,(0,0)) = (0,0)$, and 
\[f_i(t,(1,1)) = \begin{cases}(1,0), & \text{if $\#_\XX(t) < \#_\YY(t)$} \\
(0,1), & \text{if $\#_\XX(t) > \#_\YY(t))$} \\
\tfrac{1}{2}(1,0) + \tfrac{1}{2}(0,1), & \text{if $\#_\XX(t) = \#_\YY(t)$,}
\end{cases}
\]
where $\#_\XX(t)$ and $\#_\YY(t)$ denote the number of positions in $t$ where 
the third (respectively fourth) co-ordinate is 1, and $(x,y)\in \D(\XX\times \YY)$ 
denotes the point distribution supported at $(x,y)$.  This system clearly does 
give rise to information leakage, because by always requesting use of the 
resource Bob is able to (imperfectly) learn about whether Alice has requested it.

Given the specification of a channel $\C$, we are interested in the possible 
ways Alice and Bob may interact with the system, which we denote by their 
\emph{strategies}.  Clearly Alice is unable to see the messages passing between 
the system and Bob, and so her strategy at each step is represented by a 
function on the transcript of her interaction so far; similarly for Bob.

\begin{definition}
Let $\C = (\A,\B,\XX,\YY,(f_1,\ldots,f_n))$ be an $n$-IC. A \emph{classical
$\A$-strategy} (respectively $\B$-strategy) for $\C$ is a tuple
$(g_1,\ldots,g_n)$ of functions
\[g_i:(\A\times\XX)^{i-1}\rightarrow \D(\A),\]
respectively $h_i:(\B\times \YY)^{i-1}\rightarrow \D(\B)$.  Denote the sets of 
such strategies by $\Ss_\A$ and $\Ss_\B$ respectively.
\end{definition}

Having thus fixed the strategies followed by Alice and Bob, we obtain a 
probability distribution on traces of the system execution: writing 
$s_A=(g_1,\ldots,g_n)$ and $s_B=(h_1,\ldots,h_n)$ we have that the trace 
$t=((a_1,b_1,x_1,y_1),\ldots,(a_n,b_n,x_n,y_n))$ occurs with probability
\[\prod_{i=1}^n g_i(\pi_A(t_{i-1}))(a_i)h_i(\pi_B(t_{i-1}))(b_i)f_i(t_{i-1},(a_{i},b_{i}))(x_i,y_i),\]
writing $t_i$ for the $i$th prefix of $t$ and $\pi_A$ and $\pi_B$ for the 
projections onto $(\A\times\XX)^*$ and $(\B\times\YY)^*$ respectively 
representing Alice and Bob's views of the system.  We denote the trace produced 
by strategies $s_A$ and $s_B$ by the random variable $T_{s_A,s_B}$.

Now that we have defined the behaviour of the system and the parties, we are
able to talk about information flow. We assume there is some secret about which
Bob wishes to learn, which we will denote by the random variable $K$; Alice's
strategy may depend in some way on the value of $K$. The question is, how much
more does Bob know about $K$ after the interaction than before? As outlined in
Section \ref{sec:intro} there are various possible ways to measure this, so our
definition (essentially the formalism of \cite{alvim2016axioms}) is
parameterised by a `vulnerability measure' $\V$ on probability distributions.

Before he and Alice interact with the system, Bob's knowledge of the secret will 
be limited to the prior distribution of the random variable $K$; we quantify this 
knowledge by the vulnerability of this distribution according to the vulnerability 
measure $\V$, which we denote by $\V(K)$.

On the other hand, after the interaction Bob will have observed some trace $t$ 
consisting of the messages passing between him and the system, and this allows 
him to update his beliefs about the secret to the \emph{posterior} distribution 
$K|\pi_B(T)=t$ (recall that Bob is only able to observe the projection of the 
whole system trace $T$ onto $(\B\times \YY)^n$, since he does not see the 
messages passing between the system and Alice).  In quantitative terms his 
knowledge of the secret is given by $\V(K|\pi_B(T)=t)$; we call the expected 
value of this quantitiy the `posterior $\V$-vulnerability' and the expected difference 
between prior and posterior $\V$-vulnerability (that is the expected amount of 
information gained by Bob) the `$\V$-leakage' of the channel with the given 
prior distribution on $K$.

\begin{definition}\label{def:vvul}
Let $\C$ be an $n$-IC, and $K$ a random variable taking values on the set $\KK$.  
Let $\phi_A:\KK\rightarrow \Ss_\A$ and $s_B\in \Ss_\B$.  Let $\V$ be any 
vulnerability measure.  The \emph{posterior $\V$-vulnerability of $K$ under 
$(\C,\phi_A,s_B)$} is given by
\begin{align*}\mathcal{V}_\V(K,(\C,&\phi_A,s_B))\\ &= \E_{t\sim \pi_B(T_{\phi_A(K),s_B})}\V\left(K|\pi_B(T_{\phi_A(K),s_B})=t\right).\end{align*}
The \emph{$\V$-leakage of $K$ under $(\C,\phi_A,s_B)$} is given by
\[\LL_\V(K,(\C,\phi_A,s_B)) = \mathcal{V}_\V(K,(\C,\phi_A,s_B))-\V(K).\]
\end{definition}

Note that the posterior distribution $K|\pi_B(T_{\phi_A(K),s_B})=t$ is 
straightforwardly given by Bayes' theorem
\[p_{K|\pi_B(T_{\phi_A(K),s_B})}(k|t) = \frac{p_{\pi_B(T_{\phi_A(K),s_B})|K}(t|k) 
p_K(k)}{p_{\pi_B(T_{\phi_A(K),s_B})}(t)}.\]

Some important examples of vulnerability measures:
\begin{itemize}
\item Shannon entropy: $\V(K) = -H_1(K) = \sum_k p_K(k)\log(p_K(k))$.  This 
gives a measure of leakage corresponding to mutual information.
\item Min-entropy~\cite{smith2009foundations}: $\V(K) = -H_{\infty}(K) = \log \max_k
p_K(k)$. This has a natural operational interpretation, as ($\log$ of) the
multiplicative improvement in Bob's probability of guessing the value of $K$ in
one try.
\item $g$-vulnerability~\cite{alvim2012measuring}: this is a family of
vulnerability measures, parameterised by a finite set of guesses $\WW$ Bob can
make, and a `gain function' $g:\WW\times\KK \rightarrow [0,1]$ giving the
reward to Bob for making guess $w$ if the true value was $k$. Then the expected
value of Bob's multiplicative gain is given by $\V$-leakage with $\V(K) = \log
\max_w p_K(k)g(w,k)$. We may also be interested in Bob's additive gain, which is
given by $\V(K) = \max_w p_K(k)g(w,k)$ (omitting the $\log$).
\end{itemize}

The definition of $\V$-leakage can be expressed more concisely using an analogue of 
Shannon mutual information (which gives the asymptotic capacity of a binary 
symmetric channel), parametrised by the vulnerability measure $\V$: if we 
define
\[I_\V(X;Y) = \E_{y\sim Y} \V(X|Y=y) - \V(X)\]
then we have that 
\[\LL_\V(K,(\C,\phi_A,s_B)) = I_\V(K;\pi_B(T_{\phi_A(K),s_B})).\]
Note that if $\V$ is Shannon entropy then $I_\V$ is Shannon mutual information; 
this is symmetric in $X$ and $Y$ but $I_\V$ is not in general symmetric for 
other vulnerability measures.

We then define the $\V$-capacity of the channel to be the maximum possible 
$\V$-leakage over all possible secrets $K$ and all possible behaviours for Alice 
and Bob.

\begin{definition}
Let $\C$ be an $n$-IC.  The \emph{classical $\V$-capacity of $\C$} is given by
\[\LL_\V(\C) = \sup_K\sup_{\phi_A:K\rightarrow \Ss_\A, s_B\in \Ss_\B} 
\LL_\V(K,(\C,\phi_A,s_B)).\]
\end{definition}

Throughout this paper, we will consider only vulnerability measures satisfying 
three healthiness conditions, which hold for all reasonable measures and which 
we will need to use in order to prove some of our results later on (in 
particular for the proof of Theorem \ref{thm:gameequiv}).  The first
healthiness condition we will call the \emph{composition inequality}.
Informally, this says that if we have a composition of channels $X\rightarrow Y
\rightarrow Z$ then the capacity of the channel from $X$ to $Z$ is not greater
than that of those from $X$ to $Y$ and from $Y$ to $Z$. Clearly this is a
property that any sensible vulnerability measure should have.

More formally, for random variables $X,Y,Z$, we say that they form a
\emph{Markov chain}, and write $X\rightarrow Y\rightarrow Z$ if we have
$p_{X,Y,Z}(x,y,z) = p_X(x)p_{Y|X}(y|x)p_{Z|Y}(z|y)$ (note that this property is
symmetric, so that $X\rightarrow Y \rightarrow Z$ if and only if $Z\rightarrow Y
\rightarrow X$; see \cite{elementsofinformation} Section 2.8). We say that a
vulnerability measure $\V$ satisfies the composition inequality if for every
Markov chain $X\rightarrow Y\rightarrow Z$ we have
\[I_\V(X;Z) \leq \sup_{Y'|X'=Y|X}I_\V(X';Y')\]
and 
\[I_\V(X;Z) \leq \sup_{Z'|Y'=Z|Y}I_\V(Y';Z').\]

This fact for Shannon entropy vulnerability follows from the 
data-processing inequality (\cite{elementsofinformation}, Theorem 2.8.1), and for 
min-entropy is Theorem 6 of~\cite{espinoza2013min}.  Note that the first inequality
without the supremum (which would bound information flow rather than capacity) can fail for 
min-entropy vulnerability (see Example 7 of~\cite{espinoza2013min}), but both 
hold for Shannon entropy.

The second healthiness condition we will require is that the vulnerability of a 
Bernoulli random variable is (strictly) less if it is closer to uniform.  That is, 
if $\rho,\rho'\in [0,1]$ with $|\rho - 1/2| < 
|\rho' - 1/2|$ then we have
\[\V(\Ber(\rho)) < \V(\Ber(\rho')),\]
where $\Ber(\rho)$ is the Bernoulli distribution with parameter
$\rho$.

The third and final assumption we make about $\V$ is that if we have a binary
symmetric channel then the best way to use it is to send a uniformly random bit.
More concretely, we assume that if $(X,Y)$ is a binary symmetric channel with
error probability $p$ then $I_\V(X;Y)$ maximised when $X\sim \Ber(1/2)$, in
which case the posterior is $\Ber(1-p)$, so we assume
\[I_\V(X;Y) \leq \V(\Ber(1-p)) - \V(1/2).\]

A consequence of the composition inequality is that without loss of generality
we may assume that Alice employs a deterministic strategy: indeed, we may
consider her source of randomness to be a random variable $X$ (so that she
employs a deterministic strategy on $K\times X$), and then we have that
$K\rightarrow K\times X \rightarrow \pi_B(T)$ is a Markov chain, so the capacity
of the channel given by her deterministic strategy on $K\times X$ is at least
that of the original strategy. Once we have that Alice uses a deterministic
strategy we may assume that $|\KK|$ is at most the size of the set of functions
$(\A\times \XX)^{<n}\rightarrow \A$, which in particular is bounded. We can
similarly show that Bob can be assumed to use a deterministic strategy (assume
his randomness is resolved before the interaction and pick the value of the seed
leading to the greatest leakage), and so the set of possible strategies is
finite and the classical capacity of a given channel is computable.

Returning to the toy example of the fair scheduler described near the beginning 
of this section, we can easily see that this has postive $\V$-capacity under 
any vulnerability measure $\V$ satisfying the healthiness conditions.  Indeed, 
let $K\sim U(\{0,1\})$, and let Bob's strategy $s_B$ be given by $h_i(t)=1$ for 
all $i,t$ (that is, Bob always asks to use the resource).  Define strategy $s_0$ 
for Alice by $g_i(t)=0$ for all $i,t$ (never asking for the resource) and 
strategy $s_1$ by $g_i(t)=1$ for all $i,t$ (always asking for the resource).
For $k\in \{0,1\}$, let $\phi_A(k)=s_k$.

Now if $K=0$ then Bob will always receive 1 from the system.  If $n\geq 2$ then 
by fairness if $K=1$ then Bob will always receive a 0 at least once, and so we 
have that $K|\pi_B(T_{\phi_A(K),s_B})$ is a point distribution for both $K=0$ 
and $K=1$, and so 
\begin{align*}
\LL_\V(\C) &\geq \LL_\V(K,(\C,\phi_A,s_B)) \\
&= \V(\Ber(1)) - \V(\Ber(1/2)) \\
&> 0
\end{align*}
by the second healthiness condition.

If $n=1$ then if $K=1$ Bob will receive a 0 or a 1 uniformly at random.  Hence 
if he receives a 0 he can deduce with certainty that $K=1$, but if he receives a 1 then his 
posterior is that $K=0$ with probability $\tfrac{1}{2}/(\tfrac{1}{2}+\tfrac{1}{4}) 
= \tfrac{2}{3}$ and $K=1$ with probability $\tfrac{1}{3}$.  Hence we have
\begin{align*}
\LL_\V(\C) &\geq \LL_\V(K,(\C,\phi_A,s_B)) \\
&= \left(\tfrac{1}{4} \Ber(1) + \tfrac{3}{4} \Ber(2/3)\right) - \V(\Ber(1/2) \\
&> 0,
\end{align*}
again by the second healthiness condition.  Of course these lower bounds for 
$\LL_\V(\C)$ are not tight; the optimal strategy and maximum leakage will 
depend on the choice of vulnerability measure $\V$.

\subsection{Entanglement}

We give here a very brief introduction to the theory of quantum states and
quantum measurements; a more detailed introduction can be found
in~\cite{nielsen2000quantum}.

A quantum system is represented by a complex Hilbert space $\HH$ (that is, a 
complex inner product space such that the distance metric is continuous); for 
most of this work (except Section \ref{sec:sdp}) we will assume that $\HH$ is 
finite-dimensional, and so $\HH\cong \mathbb{C}^n$ for some $n$.  A \emph{qubit} 
is a system $\HH=\mathbb{C}^2$ and we write $\{\ket{0},\ket{1}\}$ an 
orthonormal basis for $\HH$ (the `standard basis vectors').

A \emph{state} of the system is a unit vector $\ket{\psi} \in \HH$ (more precisely 
this is a `pure state'; we will not need to consider mixed states in this work).
If $\HH=\HH_1\otimes \HH_2$ then we say that $\ket{\psi}\in \HH$ is 
\emph{separable} if $\ket{\psi} = \ket{\psi_1}\otimes \ket{\psi_2}$ for some 
$\ket{\psi_1} \in \HH_1$ and $\ket{\psi_2}\in \HH_2$; otherwise we say that 
$\ket{\psi}$ is \emph{entangled}.

What does it mean to make a measurement on a system $\HH$? In this work we will 
consider only \emph{projective} measurements; that is, measurements such that 
performing the measurement twice is the same as performing it once (this is 
without loss of generality since we will never care about the exact dimension 
of our Hilbert spaces and by the Naimark dilation theorem any measurement can 
be expressed as a projective measurement on a larger Hilbert space).

By an \emph{orthogonal projective measurement} over $\HH$ (hereafter just
`measurement') we mean a collection of Hermitian operators $\{E_i\}_{i\in
\mathcal{I}}$ over $\HH$, where $\mathcal{I}$ is the set of measurement
outcomes, satisfying the following properties:
\begin{enumerate}[(i)]
\item for each $i$, $E_i^2 = E_i$ (each $E_i$ is a projection),
\item $E_i E_j = 0$ for all $i\neq j$ (orthogonality), and
\item\label{it:norm} $\sum_{i\in \mathcal{I}} E_i = I$ the identity operator.
\end{enumerate}

When we apply the measurement $\{E_i\}_{i\in \mathcal{I}}$ to state
$\ket{\psi}$, we obtain result $i$ with probability $\braket{\psi|E_i|\psi}$
(where $\bra{\psi}=\ket{\psi}^* \in \HH^*$ is the dual vector to $\ket{\psi}$);
note that this is a probability distribution by condition (\ref{it:norm}).  

Some examples of measurements (on a single cubit) are measurement in the 
standard basis, $\{\ket{0}\bra{0},\ket{1}\bra{1}\}$, and measurement in the 
`Hadamard basis', $\{(\ket{0}+\ket{1})(\bra{0}+\bra{1})/2,(\ket{0}-\ket{1})
(\bra{0}-\bra{1})/2\}$, at an angle $\pi/4$ to the standard basis.

Note that measurements compose, so that if $\{E_i\}_{i\in \mathcal{I}}$ and
$\{E_j\}_{j\in \mathcal{J}}$ are projective measurements then so is
$\{E_iE_j\}_{(i,j)\in \mathcal{I}\times \mathcal{J}}$, corresponding to
measuring $\mathcal{J}$ followed by $\mathcal{I}$. Note that measurements do 
\emph{not} in general commute, so that $\{E_iE_j\}$ will give different results
to $\{E_jE_i\}$. This is the essential difference between quantum and classical
measurement, and gives rise to the famous `uncertainty principle' in quantum
mechanics. 

\section{Entangled channel capacity}\label{sec:entangledcap}

\subsection{Definition}\label{sec:entangleddef}

We will now consider information flow in the situation in which Alice and Bob
may share entanglement. This means that it is no longer possible to consider
their strategies entirely separately: they share some entangled state
$\ket{\psi}$, and at each step make measurements on their own part of the state
(which may depend on the history of their own communication with the system up
to that point), and choose a message to send to the system according the the
result of the measurement. Alice's choice of measurements, but not Bob's, may
also depend on the value of the secret $K$. Note that without loss of generality
we may assume that each measurement consists of one projection for each element
of $\A$ (respectively $\B$), since any post-processing of the measurement result
into a (possibly random) choice of message can be incorporated into the
measurement.

\begin{definition}
Let $\C$ be an $n$-IC, and $K$ a random variable taking values on the set $\KK$.  
A \emph{quantum joint strategy} for $\C$ is a pure state $\ket{\psi}$ in a 
finite-dimensional complex Hilbert 
space $\HH=\HH_\A\otimes \HH_\B$, and sets $\{A^{k,t}\}$ and $\{B^{t'}\}$ such 
that 
\begin{enumerate}[(i)]
\item for every $k\in \KK$ and every $t\in (\A\times\XX)^{i}$ with $0\leq i < n$, 
$A^{k,t}=\{A^{k,t}_a\}_{a\in\A}$ is a measurement over $\HH_\A$, and 
\item for every $t'\in (\B\times \YY)^i$ with $0\leq i < n$, 
$B^{t'}=\{B^{t'}_b\}_{b\in\B}$ is a measurement over $\HH_\B$.
\end{enumerate}
Denote the space of such strategies by $\Ss^*_{\C,K}$.
\end{definition}

We again denote the trace produced by strategy $s$ with secret $k$ by the random
variable $T_{s,k}$. What is the probability that $T_{s,k}$ takes the value
$t=((a_1,b_1,x_1,y_1),\ldots,(a_n,b_n,x_n,y_n))$? Whereas before in the
classical case this was given by the product of the relevant classical
probabilities corresponding to the execution $t$ from the functions defining the
strategies of Alice and Bob and the behaviour of the machine, now for Alice and
Bob we must find the probability that the corresponding sequences of
measurements result in the correct outcomes. This is given by the norm on
$\ket{\psi}$ of the product of the corresponding projections; on the other hand,
since the system itself is purely classical its probability is still given by
multiplying the relevant probabilities.

Denote by $A^k_t$ and $B_t$ the 
projections corresponding to Alice and Bob taking the actions corresponding to 
trace $t$ at each step; that is
\begin{align*}
A^k_t &= A^{k,((a_1,x_1),\ldots,(a_{n-1},x_{n-1}))}_{a_n}
A^{k,((a_1,x_1),\ldots,(a_{n-2},x_{n-2}))}_{a_{n-1}} \\
&\qquad\qquad \ldots 
A^{k,((a_1,x_1))}_{a_2} 
A^{k,\emptyset}_{a_1} \\
B_t &= B^{((b_1,y_1),\ldots,(b_{n-1},y_{n-1}))}_{b_n} 
B^{((b_1,y_1),\ldots,(b_{n-2},y_{n-2}))}_{b_{n-1}} \\
&\qquad\qquad 
\ldots B^{((b_1,y_1))}_{b_2} B^{\emptyset}_{b_1}.
\end{align*}

Writing $t_i$ for the $i$th prefix of $t$ as before, we have that the
probability that $T_{s,K}$ takes the value $t$ is given by 
\[\braket{\psi|A^k_t\otimes B_t|\psi} \prod_{i=1}^n f_i(t_i)(x_i,y_i),\]
where the $f_i$ are the functions specifying the channel behaviour from 
Definition \ref{def:chan}.

As in the classical case, we then say that the information leakage from Alice to
Bob $\LL_\V(K,(\C,s))$ is given by the increase in $\V$-vulnerability from the
prior to Bob's posterior distribution after the interaction, for our preferred
choice of vulnerability measure $\V$.

\begin{definition}\label{def:entcap}
Let $\C$ be an $n$-IC.  The \emph{entangled $\V$-capacity of $\C$} is given by 
\[\LL^*_\V(\C) = \sup_K \sup_{s\in \Ss^*_{\C,K}} \LL_\V(K,(\C,s)),\]
where $\LL_\V(K,(\C,s))$ is defined equivalently to Definition \ref{def:vvul}.
\end{definition}

Trivially $\LL_\V^*(\C) \geq \LL_\V(\C)$ for any channel $\C$.  We will write 
$\Delta^*_\V(\C)$ for the `quantum advantage'
\[\Delta_\V^*(\C) = \LL_\V^*(\C) - \LL_\V(\C).\]
We will say that $\C$ is a \emph{purely quantum channel} if $\LL_\V^*(\C) 
> \LL_\V(\C) = 0$.

With this as our central definition, in the remainder of this paper we will 
investigate some of its fundamental questions, in particular: is it possible to 
have $\Delta^*_\V(\C) > 0$?  Is it possible to have $\LL_\V(\C) =0$ but 
$\LL^*_\V(\C) > 0$?  Given a channel $\C$, can we compute $\LL^*_\V(\C)$?  Given 
that (as we shall see) the answer to the previous question is `no', can we at 
least get some bounds on it?

\subsection{Non-local games}\label{sec:nonlocal}

The key technical ingredient for many of the results of this paper is the
observation that the entangled capacity of interactive channels has a close
connection with the theory of \emph{non-local games}. This is a formalism that
highlights and in some sense allows us to measure the inherently `contextual'
nature of quantum mechanics: that is, that it is possible for two parties
sharing entanglement to accomplish tasks that would be impossible for separated
parties under any purely local theory of physics.

The basic setup is that we have two players, Alice and Bob, playing a
(co-operative) game with a referee. The referee begins by sending Alice and Bob
a message drawn (probabilistically) from finite sets $\A$ and $\B$ respectively.
Alice and Bob must then respond with messages from sets $\XX$ and $\YY$
respectively. The referee then determines according to a specified function $D$
whether Alice and Bob have won or lost the game; we are interested in the highest
probability with which Alice and Bob can win, which we call the `value' of the
game. (One could also consider games with more players or more rounds, but we
will not need to for this work.)

\begin{definition}
A \emph{two-player one-round non-local game} is a tuple 
$\gG=(\A,\B,\XX,\YY,D,\mu)$, where $\A,\B,\XX$ and $\YY$ are finite sets, 
$\mu\in \D(\A\times \B)$ is some probability distribution and
$D:\A\times \B \times \XX \times \YY \rightarrow \{0,1\}$ is the \emph{decision 
function}, determining whether Alice and Bob are considered to have won or lost 
the game.

A \emph{classical} strategy $s$ for $\gG$ comprises a pair of functions 
$f:\A\rightarrow \D(\XX)$ and $g:\B\rightarrow \D(\YY)$.  Write
\[\val(\gG,s) = \sum_{a,b,x,y} \mu(a,b) f(a,x) g(b,y) D(a,b,x,y)\]
for the win probability of strategy $s$, and 
\[\val(\gG) = \sup_s \val(\gG,s),\]
the \emph{classical value} of $\gG$.
\end{definition}

It is easy to show that in fact Alice and Bob's optimal win probability can be
obtained with purely deterministic strategies, so that without loss of 
generality we may assume $f(a)(x),f(b)(y)\in\{0,1\}$ for all $a,b,x,y$.

What if Alice and Bob are given access to entanglement? As for channels, we allow 
Alice and Bob to share some quantum state $\ket{\psi}\in \HH_\A\otimes \HH_\B$.  
The strategy must specify a measurement taking values on $\XX$ for each message 
Alice could receive; similarly for Bob.

\begin{definition}
Let $\gG=(\A,\B,\XX,\YY,D,\mu)$ be a game. A \emph{quantum strategy} for $\gG$
is a pure state $\ket{\psi}$ in a finite-dimensional complex Hilbert space
$\HH=\HH_\A\otimes \HH_\B$, and sets $\{A^a\}$ and $\{B^b\}$ such that for every
$a\in \A$, $A^a=\left\{A^a_x\right\}_{x\in \XX}$ is a measurement over $\HH_\A$, and for
every $b\in \B$, $B^b=\left\{B^b_y\right\}_{y\in \YY}$ is a measurement over $\HH_\B$.

For strategy $s$ as above, let
\[\val(\gG,s) = \sum_{a,b,x,y} \mu(a,b) \braket{\psi|A^a_x\otimes B^b_y|\psi} 
D(a,b,x,y).\]
Then the \emph{entangled value} of $\gG$ is given by 
\[\val^*(\gG) = \sup_s \val(\gG,s).\]
\end{definition}

The original example of a non-local game is the \emph{CHSH
game}~\cite{clauser1969chshgame}, which we denote $\gG_\CHSH$. In this game, the
messages sent and received by Alice and Bob each consist of a single bit. The
judge sends each player a uniformly random bit $a,b$; they each reply with a
single bit $x,y$. The players' goal is to arrange that if $a= b=1$ then $x$ and
$y$ are \emph{different}, and otherwise $x$ and $y$ are \emph{equal}. Formally,
we have $\A=\B=\XX=\YY=\{0,1\}$, $\mu=U(\A\times \B)$ the uniform distribution
and $D(a,b,x,y) = 1_{ab=x\oplus y}$.

It is fairly easy to see that if Alice and Bob are restricted to classical 
strategies then they cannot do better than just both always returning 0 (say).  
Since $a=b=1$ occurs only with probability $1/4$, this means that they win with 
probability $3/4$.

On the other hand, as we see in Proposition \ref{prop:CHSH}, if Alice and 
Bob are given access to entangled strategies then they can win with probability 
$\cos^2(\pi/8) \approx 0.85 > 0.75$. 

\begin{proposition}\label{prop:CHSH}
We have
\[\val^*(\gG_\CHSH) \geq \cos^2(\pi/8) \approx 0.85 > \val(\gG_\CHSH) = 3/4.\]
\end{proposition}
\begin{proof}
It is easy to check that the optimal classical strategy is for Alice and Bob to 
always send $0$, which has win probability 3/4.  We exhibit an 
entangled strategy with win probability $\cos^2(\pi/8)$.  Let 
$\HH_\A=\HH_\B=\mathbb{C}^2$ and 
$\ket{\psi} = (\ket{0}\otimes \ket{0} + \ket{1}\otimes \ket{1})/\sqrt{2}$.  For 
$\theta\in [-\pi,\pi]$ we will write $\ket{\theta} = \cos\theta \ket{0} + 
\sin\theta \ket{1}$, and $[\theta] = \ket{\theta}\bra{\theta}$.  Let 
$A^0_0 = [0]$, $A^0_1 = [\pi/2]$, $A^1_0 = [\pi/4]$, $A^1_1 = [3\pi/4]$, 
$B^0_0 = [\pi/8]$, $B^0_1=[5\pi/8]$, $B^1_0=[-\pi/8]$ and $B^1_1=[3\pi/8]$.  One 
can check that this strategy has win probability $\cos^2(\pi/8)$.
\end{proof}

This game (implemented with Alice and Bob sufficiently separated as to preclude
communication between them) has been used to show experimentally that despite
Einstein's qualms the behaviour of the universe is in fact inherently non-local,
since the players can obtain a winning strategy higher than that attainable in
any purely local theory.

\subsection{Quantum advantage}

We now show, using a channel derived from the CHSH game, that it is 
possible for entanglement to increase channel capacity.  Essentially we  
define a channel which plays the CHSH game with Alice and Bob, and if they win 
rewards them by transmitting a single bit of information. Since entanglement 
increases the probability with which they can win the game, it increases the 
capacity of the channel.

Concretely, define $\C_\CHSH$ to be a two-round interactive channel with 
$\A=\{0,1\}\times \{0,1\}$ and $\B = \XX =\YY = \{0,1\}$.  Define 
$f_1(a,b) = U(\XX\times \YY)$ (that is, Alice and Bob's first round inputs are 
ignored), and 
\begin{align*}f_2((a_1,b_1,x,y),((a_2,&a'_2),b_2)) \\ &= 
\begin{cases} (0,a'_2) &\text{if $a_2\oplus b_2 = xy$} \\
(0,U(\{0,1\})) &\text{otherwise}.
\end{cases}
\end{align*}

\begin{theorem}\label{thm:quantad}
Let $\V$ be a vulnerability measure.  We have
\begin{align*}
\LL_\V^*(\C_\CHSH) &\geq \V(\Ber((1+\cos^2(\pi/8))/2) - \V(\Ber(1/2)) \\
&> \LL_\V(\C_\CHSH) = \V(\Ber(7/8)) - \V(\Ber(1/2)).
\end{align*}
\end{theorem}
\begin{proof}
For the lower bound on $\LL_\V^*(\C_\CHSH)$, let $\HH_\A=\HH_\B=\mathbb{C}^2$
and $\ket{\psi} = (\ket{0}\otimes \ket{0} + \ket{1}\otimes \ket{1})/\sqrt{2}$.
Let $K$ be uniformly distributed on $\{0,1\}$. Let $A^{k,()}=B^{()} =
\{I,0,\ldots\}$. Let $A^{k,(a_1,x)}_{(a,k)} = A^x_a$ and $B^{(b_1,y)}_b = B^y_b$
from the proof of Proposition \ref{prop:CHSH}, and $A^{k,(a_1,x)}_{(a,1-k)} =
0$. Now since $a_2\oplus b_2 = xy$ with probability $\cos^2(\pi/8)$ we have that
$y_2=k$ with probability $\cos^2(\pi/8) +
(1-\cos^2(\pi/8))/2=(1+\cos^2(\pi/8))/2$. Conditional on observing $y_2$, Bob's
posterior probability that $k=y_2$ is $((1+\cos^2(\pi/8))/2)\cdot (1/2)/(1/2)$
by Bayes' theorem (see the formula immediately below Definition \ref{def:vvul}),
so the posterior vulnerability is $\V(\Ber((1+\cos^2(\pi/8))/2))$, as required.

For the upper bound on $\LL_\V(\C_\CHSH)$, we have that without loss of 
generality Alice and Bob employ deterministic strategies, and so it is a finite 
check to establish that their optimal strategy is $a_2=b_2=0$ and $a_2'=k$ with 
$K\sim U(\{0,1\})$, giving leakage $\V(\Ber(7/8)) - \V(\Ber(1/2))$ as required.
The strict inequality follows from the second healthiness condition on $\V$.
\end{proof}

\section{Purely quantum channels}\label{sec:purelyquant}

In this section, we will show that it is not possible for entanglement to 
increase the capacity of a channel with zero classical capacity.  In fact we do 
this by showing the slightly stronger result that a zero-classical-capacity 
channel has zero capacity even if Alice and Bob are allowed strategies involving 
\emph{any} `non-signalling' correlations---that is, such that Bob's choice at a 
particular stage does not in itself convey information for him, and similarly 
for Alice (recall that all correlations resulting from entanglement are 
non-signalling; but not all non-signaling correlations can be produced using 
entanglement).  In mathematical terms this corresponds to saying that the 
marginal distribution on Bob's next action (respectively Alice's) is independent 
of the history of Alice's (respectively Bob's) part of the interaction.

\begin{definition}\label{def:genstrat}
Let $\C$ be an $n$-IC, and $K$ a random variable.  A \emph{generalised strategy} 
$s$ for $\C$ is a tuple $(g_1,\ldots,g_n)$ of functions
\[g_i:\KK\times (\A\times\B\times\XX\times\YY)^{i-1} \rightarrow \D(\A\times \B).\]
We say $s$ is \emph{non-signalling} if
\begin{enumerate}[(i)]
\item\label{cond:atob} for every $k,k'\in \KK$ and $t,t' \in (\A\times\B\times\XX\times\YY)^n$ 
with $\pi_B(t)=\pi_B(t')$, we have $\forall b\in \B$
\[\sum_{a\in\A}g_i(k,t_i)(a,b) = \sum_{a\in\A}g_i(k',t'_i)(a,b)\]
for every $i$, and
\item for every $k\in \KK$ and $t,t' \in (\A\times\B\times\XX\times\YY)^n$ 
with $\pi_A(t)=\pi_A(t')$, we have $\forall a\in \A$
\[\sum_{b\in\B}g_i(k,t_i)(a,b)=\sum_{b\in\B}g_i(k,t'_i)(a,b)\]
for every $i$.
\end{enumerate}
\end{definition}

For a vulnerability measure $\V$ we define $\V$-leakage under strategy $s$ as 
before, and the supremum of such leakage under all non-signalling strategies as 
the non-signalling $\V$-capacity, which we denote $\LL_\V^\ns(\C)$.

Every quantum joint strategy is a non-signalling strategy, with 
\begin{align*}
g_i(k,t_i)(a,b) = \braket{\psi|(A^{k,\pi_A(t_{i-1})}_a  A^k_{t_{i-1}})\otimes 
(B^{\pi_B(t_{i-1})}_b B_{t_{i-1}})|\psi}.
\end{align*}

Hence for any channel $\C$ we have 
\begin{equation}\label{eq:nonsigineq}
\LL_\V^\ns(\C) \geq \LL_\V^*(\C) \geq \LL_\V(\C).
\end{equation}

Note that this inequality can be strict: for example, it is easy to show that
non-signalling correlations allow Alice and Bob to win the CHSH game with
probability 1, and so (as we will see in Theorem \ref{thm:gameequiv} in Section
\ref{sec:noncomp} below) we have $\LL_\V^\ns(\C_\CHSH) \geq
\V(\Ber(1))-\V(Ber(1/2)) > \LL_\V^*(\C_\CHSH)$.

The reason for considering this broader class of strategies is that they can be 
analysed in an abstract linear-algebraic manner.  Define the set 
$\DD_\A=[(\A\times \XX)^{<n} \rightarrow \A]$, the set of functions 
$(\A\times \XX)^{<n} \rightarrow \A$, and similarly 
$\DD_\B=[(\B\times \YY)^{<n} \rightarrow \B]$.  A channel gives a map 
\[c:\DD_\A\times \DD_B \rightarrow \D(\B\times \YY)^n,\] where $c(f,g)$ is the 
probability distribution on Bob's traces if Alice behaves according to $f$ and 
Bob behaves according to $g$.

To accommodate probabilistic behaviour by Alice and Bob, we extend the function 
$c$ by linearity to a linear map
\[\C:\R\DD_\A \otimes_\R \R\DD_\B \rightarrow \R(\B\times \YY)^n,\]
where $\R X$ is the free real vector space over the set $X$.

Note that $\C$ is `trace-preserving', where the trace of a vector $v$ is the 
trace of the linear map $u \mapsto \langle u, v\rangle v$, or in more concrete terms 
$\tr(\sum_i c_i e_i) = \sum_i c_i$, where the $e_i$ are the canonical basis 
vectors.

Observe that $\R\DD_\A$ is canonically isomorphic to $[(\A\times\XX)^{<n} 
\rightarrow \R\A]$, and similarly $\R\DD_\B$ to $[(\B\times\YY)^{<n} \rightarrow 
\R\B]$, which extend to a canonical isomorphism between $\R\DD_\A \otimes_\R 
\R\DD_\B$ and $[(\A\times\B\times\XX\times\YY)^{<n}\rightarrow \R(\A\times \B)]$.
By currying Definition \ref{def:genstrat} and observing that $\D(\A\times \B) 
\subset \R(\A\times\B)$ we see that a generalised strategy corresponds to a map
\[s:\KK\rightarrow \R\DD_\A \otimes_\R \R\DD_\B.\]

The payoff to all this is that we get a clean characterisation of the property 
that the marginal distribution on Bob's strategies cannot depend on the value of 
the secret, which turns out to suffice for the theorem.

\begin{lemma}\label{lem:nonsig}
Let $s$ be a non-signalling generalised strategy. Then we have that
$tr_\A(s(k))$ is constant for all $k$, where $\tr_\A$ is the `partial trace'
function
\[\tr_\A:\R\DD_\A \otimes_\R \R\DD_\B \rightarrow \R\DD_\B,\]
the linear function generated by $f\otimes g \mapsto g$ for $f\in \DD_\A, 
g\in \DD_\B$.
\end{lemma}
\begin{proof}
Let $k,k'\in \KK$.  By condition (\ref{cond:atob}) of 
Definition \ref{def:genstrat} we have for every trace prefix $t_i$ and 
every $b\in \B$ that 
\begin{equation}\label{eq:nonsig}
\sum_{a\in \A} \langle s(k)(t_i), a\otimes b \rangle =  
\sum_{a\in \A} \langle s(k')(t_i), a\otimes b \rangle
\end{equation}
(where we define $(f\otimes g)(t_i)=f(\pi_A(t_i))\otimes g(\pi_B(t_i))$ and 
extend by linearity).

Now trivially we have
\[s(k)=\sum_{f\in \DD_\A, g\in \DD_\B} \langle s(k), f\otimes g\rangle f\otimes 
g\]
and so 
\begin{align*}
s(k)(t_i) &= \sum_{f,g} \langle s(k), f\otimes g\rangle (f\otimes g)(t_i) \\
&= \sum_{f,g} \langle s(k),f\otimes g\rangle f(t_i)\otimes g(t_i)
\end{align*}
(dropping the $\pi_A$ and $\pi_B$ for conciseness).

Hence
\begin{align}
\sum_{a\in \A} \langle s(k)(t_i&), a\otimes b \rangle \nonumber\\ 
&= \sum_a \sum_{f,g} 
\langle s(k),f\otimes g\rangle \langle f(t_i)\otimes g(t_i), 
a\otimes b\rangle \nonumber\\
&= \sum_{f,g} \langle s(k),f\otimes g\rangle \langle f(t_i), {\textstyle \sum}_a a\rangle 
\langle g(t_i), b\rangle \nonumber\\
&= \sum_{f,g} \langle s(k),f\otimes g\rangle \langle g(t_i), b\rangle,\label{eq:treq}
\end{align}
since $f\in \DD_\A$ and so $f(t_i) \in \A$ so $\langle f(t_i), \sum_a a\rangle 
=1$.

Note that another characterisation of $\tr_\A$ is 
\[\tr_\A(\xi) = \sum_{f,g} \langle \xi, f\otimes g\rangle g\]
(trivially true on the basis vectors $f'\otimes g'$ and hence by linearity true 
in general), and hence we have
\begin{align*}
\tr_\A(s(k))(t_i) &= \sum_{f,g} \langle s(k),f\otimes g\rangle g(t_i).
\end{align*}
Combining this with (\ref{eq:treq}) and (\ref{eq:nonsig}) gives that 
\[\langle \tr_\A(s(k))(t_i), b\rangle = \langle \tr_\A(s(k'))(t_i), b\rangle\]
for all $b$, and hence that $\tr_\A(s(k))(t_i)=\tr_\A(s(k'))(t_i)$ for all $t_i$ and 
so $\tr_\A(s(k))=\tr_\A(s(k'))$, as required.
\end{proof}

Note that our `partial trace' $\tr_\A$ is indeed a classical analogue of the 
familiar partial trace from quantum information theory.

We are now ready to prove the main theorem of this section, that non-signaling 
strategies cannot increase the capacity of a channel with zero classical 
capacity.

\begin{theorem}\label{thm:nonsig}
Let $\C$ be an $n$-IC and $\V$ a vulnerability measure.  Then
\[\LL_\V(\C) = 0 \Rightarrow \LL_\V^\ns(\C) = 0.\]
\end{theorem}
\begin{proof}
Let $\C$ be an $n$-IC with $\LL_\V^\ns(\C) > 0$, so in particular there exists 
a non-signalling strategy $s:\{0,1\} \rightarrow \R\DD_\A\otimes_\R \R\DD_\B$ 
with $\C(s(0)) \neq \C(s(1))$.

We claim that there must exist $f,f' \in \DD_\A$ and $g\in \DD_\B$ such that 
$c(f,g)\neq c(f',g)$ (equivalently $\C(f\otimes g) \neq \C(f'\otimes g)$), and 
hence $\C$ has positive classical capacity.

Indeed, supposing the contrary for each $g\in \DD_\B$ there exists $y_g \in 
\D(\B\times \YY)^n$ such that $\C(f\otimes g) = y_g$ for all $f\in \DD_\A$.  
Write
\begin{align*}
s(0) &= \sum_{f\in \DD_\A, g\in \DD_\B} c_{f,g} f\otimes g \\
s(1) &= \sum_{f\in \DD_\A, g\in \DD_\B} c'_{f,g} f\otimes g .
\end{align*}

Since $s$ is non-signalling, by Lemma \ref{lem:nonsig} for all $g\in \DD_\B$ we 
have 
\[\sum_{f\in \DD_\A} c_{f,g} = \sum_{f\in \DD_\A} c'_{f,g}.\]

But then
\begin{align*}
\C(s(0)) &= \C\left(\sum_{f,g} c_{f,g} f\otimes g\right) \\
&= \sum_g \left(\sum_f c_{f,g} \right) y_g \\
&=\sum_g \left(\sum_f c'_{f,g} \right) y_g \\
&= \C(s(1)),
\end{align*}
a contradiction.
\end{proof}

Combining Theorem \ref{thm:nonsig} with inequality (\ref{eq:nonsigineq}) gives 
the corresponding result for entangled capactiy.

\begin{corollary}\label{cor:purelyquant}
Let $\C$ be an $n$-IC and $\V$ a vulnerability measure.  Then
\[\LL_\V(\C) = 0 \Rightarrow \LL_\V^*(\C) = 0.\]
\end{corollary}

\section{Non-computability of entangled capacity}\label{sec:noncomp}

In this section we will show that the problem of computing the entangled 
capactity of a given channel, even approximately, is RE-complete---that is, as 
hard as the halting problem.

The key ingredient is Theorem \ref{thm:mipstar}, the recent breakthrough result
of Ji, Natarajan, Vidick, Wright and Yuen which shows that computing the
entangled value of a given non-local game is RE-complete. This was formerly a
notorious open problem, because a proof of undecidability would resolve in the
negative Tsirelson's problem (asking whether the `commuting operator'
model---see Section \ref{sec:sdp}---could produce the same correlations as the
tensor product model described in Section \ref{sec:nonlocal}), which was known
to be equivalent~\cite{junge2011connes} to a famous open problem in the theory
of operator algebras, the `Connes embedding problem', open since
1976~\cite{connes1976classification}.

\begin{theorem}[\cite{ji2020mipre}, Theorem 12.7]\label{thm:mipstar}
The problem of approximating $\val^*(\gG)$ for a given $\gG$ is 
RE-complete.  More precisely, the problem of determining whether a given 
$\gG$ has $\val^*(\gG)=1$ or $\val^*(\gG) \leq 1/2$, 
given that one of these is the case, is RE-complete.
\end{theorem}

In order to apply this to entangled channel capacity, we show how to associate
to any non-local game $\gG$ a channel $\C_\gG$ such that the entangled capacity
of $\C_\gG$ and the entangled value of $\gG$ are related by an explicit formula.
This shows that the problem of computing entangled values of games is reducible
to the problem of computing entangled capacity of channels, which shows that the
latter is also RE-complete.

Informally, for a game $\gG$, we will define $\C_\gG$ as the channel that does 
the following:
\begin{enumerate}
\item Send messages $x\in \XX$ and $y\in \YY$ to Alice and Bob respectively, drawn according 
to the distribution $\mu$
\item Receive messages $a\in \A$ and $b\in \B$ from Alice and Bob respectively, 
together with a bit $u\in \{0,1\}$ from Alice
\item If $D(x,y,a,b)=1$ send the bit $u$ to Bob; otherwise send Bob a uniformly 
random bit.
\end{enumerate}

\begin{definition}
For a game $\gG=(\A,\B,\XX,\YY,D,\mu)$, define the 2-round abstract interactive channel $\C_\gG$ to 
comprise the tuple of finite sets $(\A\times \{0,1\},\B,\XX,\YY\times\{0,1\})$ 
and the functions $(f_1,f_2)$, where 
\[f_1((a,u),b)(x,(y,v)) = \mu(x,y) 1_{v=0},\]
and 
\begin{align*}f_2(((a_1,u_1),b_1,&x_1,(y_1,v_1)),((a_2,u_2),b_2)) 
= \\ &(x_0,(y_0,u_1\oplus X (1-D(x_1,y_1,a_2,b_2)))),
\end{align*}
where $X\sim \text{Ber}(1/2)$ and $x_0$ and $y_0$ are arbitrary fixed elements 
of $\XX$ and $\YY$ respectively.
\end{definition}

The main theorem of this section is that the entangled capacity of $\C_\gG$ is 
that given by the obvious strategy of setting $u$ equal to the value of the 
secret and following an optimal strategy for $\gG$.\footnote{Or rather strictly 
speaking the supremum of strategies corresponding to near-optimal strategies for 
$\gG$.}

\begin{theorem}\label{thm:gameequiv}
Let $\gG$ be a game, and $\V$ a vulnerability measure.  Then
\[\LL_\V^*(\C_\gG) = \V(\Ber((1+\val^*(\gG))/2))-\V(\Ber(1/2)).\]
\end{theorem}
\begin{proof}
We prove separately matching upper and lower bounds for $\LL_\V^*(\C_\gG)$.  The 
lower bound is trivial: given an entangled strategy for $\gG$ achieving win 
probability $p=\val^*(\gG)-\epsilon$, set $K\sim \text{Ber(1/2)}$ and have 
Alice and Bob execute the strategy for $\gG$, with Alice sending the value of $K$ 
as her additional bit $u$.

Conditional on observing the value $v$, the posterior 
probability that $K=v$ is \[\frac{p/2+(1-p)/4}{p/2 + (1-p)/2}=(1+p)/2.\]  Hence 
the posterior vulnerability is $\V(\Ber((1+p)/2))$, and this strategy 
achieves leakage $\V(\Ber((1+p)/2))-\V(\Ber(1/2))$, as required.

For the upper bound, let $s$ be a strategy for $\C_\gG$ achieving leakage 
$l$.  Let the random variable $U$ be the bit sent by 
Alice, $V$ the bit received by Bob and $W=D(x,y,a,b)$, the event that they `win' 
the game.

By considering each $s(k)$ as a strategy for $\gG$, we have that 
\[p=\max_k p_{W|K}(1|k) \leq \val^*(\gG).\]

Now, we have that $K\rightarrow (U,W) \rightarrow V$ is a Markov chain, but since 
the event that $W=1$ may depend on $K$ in an uncontrolled way we do not have 
that $K\rightarrow U \rightarrow V$ is a Markov chain.  Our 
strategy will be to show that the dependence of $V$ on $K$ can be `factored 
through' a random variable $U'$ so that $K\rightarrow U' \rightarrow V$ is a 
Markov chain and $(U',V)$ is a binary symmetric channel with error probability 
$(1-p)/2$.

Indeed, for $k\in \KK$ we must have $V|K=k \sim \Ber(\rho_k)$ for some $\rho_k$ 
(this follows just from that fact that $V|K=k$ is a $\{0,1\}$-valued random 
variable).  In particular, we have 
\begin{align}
\rho_k &=
p_{V|K}(1|k)\nonumber\\ &= p_{W|K}(1|k)p_{U|K}(1|k) \nonumber \\
&\qquad\qquad +p_{W|K}(0|k)\frac{p_{U|K}(1|k)+p_{U|K}(0|k)}{2} \nonumber\\
&=p_{W|K}(1|k)p_{U|K}(1|k) + \frac{p_{W|K}(0|k)}{2} \nonumber\\
&\in \left[\frac{p_{W|K}(0|k)}{2},1-\frac{p_{W|K}(0|k)}{2}\right].\label{eq:rhok}
\end{align}

Now, putting \[U'|K=k\sim \Ber\left(\frac{p+2\rho_k-1}{2p}\right)\] and 
\[V=U'\oplus\Ber\left(\frac{1-p}{2}\right)\] independently of $K$, we have that 
$V|K=k\sim \Ber(\rho_k)$ for all $k$, as required (since one can check that the 
xor of Bernoulli random variables with parameters $(p+2\rho_k-1)/2p$ and $(1-p)/2$ 
is a Bernoulli random variable with parameter $\rho_k$).  Note that by 
(\ref{eq:rhok}) and the fact that $p_{W|K}(0|k)=1-p_{W|K}(1|k) \geq 1-p$, we have 
that $0\leq (p+2\rho_k-1)/2p \leq 1$.

Now since $K\rightarrow U' \rightarrow 
V$ forms a Markov chain and $(U',V)$ is a binary symmetric channel with error 
probability $(1-p)/2$, by the composition inequality for 
$\V$ we have that
\begin{align*}l=I_\V(K;V) &\leq \sup_{V'|U''=V|U'}I_\V(U'';V') \\
&\leq \V(\text{Ber}((1+p)/2)) - \V(\text{Ber}(1/2))\\
&\leq \V(\text{Ber}((1+\val^*(\gG))/2)) - \V(\text{Ber}(1/2)),\end{align*}
as required.
\end{proof}

Note that Theorem \ref{thm:quantad} is a special case of the lower bound in 
Theorem \ref{thm:gameequiv}, with $\gG=\gG_\CHSH$.

Combining Theorem \ref{thm:gameequiv} with Theorem \ref{thm:mipstar} gives the result that 
computing the entangled capacity of a given channel is undecidable.  Note that 
the gapped problem is clearly in RE, because we can explicitly enumerate 
entangled strategies and accept if we find one with capacity above the lower 
threshold.

\begin{theorem}\label{thm:capnoncomp}
Let $\V$ be a vulnerability measure.  The problem of determining whether a given 
channel $\C$ has $\LL_\V^*(\C) \geq \delta^+_\V$ or $\LL_\V^*(\C) \leq 
\delta^-_\V$, given that one of these is the case, is RE-complete, where 
$\delta^+_\V > \delta^-_\V$ are the constants
\begin{align*}
\delta^+_\V &= \V(\text{Ber}(1)) - \V(\text{Ber}(1/2))\\
\delta^-_\V &= \V(\text{Ber}(3/4) - \V(\text{Ber}(1/2)).
\end{align*}
\end{theorem}

Note, for instance, that for min-entropy vulnerability we have $\delta^+_\V = 
\log_2 2 = 1$ and $\delta^-_\V = \log_2 (3/2) \approx 0.58$, and for Shannon 
entropy vulnerability we have $\delta^+_\V = H_2(1/2)-H_2(1)=1$ and $\delta^-_\V 
= H_2(1/2) - H_2(3/4)\approx 0.19$.

\section{SDP upper bounds}\label{sec:sdp}

In this section, we will show how upper bounds for entangled capacity, under the
min-entropy vulnerability measure $\V=-H_\infty$, may be obtained using
semidefinite programming. This is by analogy to a similar
method~\cite{navascues2008convergent} for non-local games.  We first introduce 
semidefinite programming and outline the technique of~\cite{navascues2008convergent}, 
and then show how it may be adapted to obtain bounds on entangled min-entropy channel 
capacity.

\subsection{The SDP hierarchy}

\emph{Semidefinite programming} (SDP)~\cite{wolkowicz2012handbook} is a technique from numerical
optimisation. A semidefinite progamming problem (for us; there are many
equivalent formulations) is specified by an \emph{objective} matrix $A$ and
\emph{constraints} given by matrices $A_1,\ldots,A_n$ and scalars $a_1,
\ldots,a_n$, and consists of the following optimisation:
\begin{align*}
\text{maximise } &A \bullet M \\
\text{subject to } & M \succcurlyeq 0 \\
&A_i \bullet M = a_i \qquad \forall i,
\end{align*}
where $\bullet$ is the Frobenius product $M\bullet N = \sum_{ij} M_{ij}N_{ij}$,
and $M\succcurlyeq 0$ means that $M$ is \emph{positive semidefinite}---that is,
$M$ is Hermitian with all its eigenvalues non-negative; equivalently, $M$ is
Hermitian and we have $v^*Mv \geq 0$ for any vector $v$. We say that a problem
is \emph{feasible} if there exists an $M$ satisfying the constraints (ignoring
the objective). The benefit of formulating a problem in this way is that the
optimisation can be performed (to specified precision) in polynomial time, and
indeed in a way which is usually efficient in practice.  Note that it is easy 
to show that within this form we may introduce additional scalar variables 
together with arbitrary linear equality or inequality constraints with the
entries of $M$, and we will allow ourselves to do this freely below.

The \emph{SDP hierarchy}, introduced in the seminal 
paper~\cite{navascues2008convergent}, uses SDP to obtain an infinite sequence of 
stronger and stronger constraints on quantum behaviours, which importantly are 
tight in the limit: that is, if a behaviour is not quantum then this will be 
detected at some finite level of the hierarchy.  The catch is that the notion 
of `quantum' used is not the usual one of Alice and Bob having their own parts 
of the system, but rather that they share some infinite-dimensional Hilbert 
space, and the only constraint is that all of Alice's measurements should 
commute with all of Bob's.  This is called the `commuting operator' model, and 
is clearly a generalisation of the usual tensor product model (since if the 
system takes the form $\HH_\A\otimes \HH_\B$ with Alice and Bob's measurements 
being only on $\HH_\A$ and $\HH_\B$ respectively then clearly their measurements 
commute); Tsirelson's conjecture asserted that the two models were equivalent, 
but this was refuted as a consequence of the recent result $\text{MIP}^*=\text{RE}$~\cite{ji2020mipre}.

More concretely, a \emph{behaviour} means a collection of probability
distributions $P(x,y)\in \D(\A\times \B)$ for each $(x,y)\in (\XX\times\YY)$,
for some finite sets $\A,\B,\XX,\YY$. We want to determine whether or not there
exists some complex Hilbert space $\HH$ (not necessarily finite-dimensional), a
state $\psi\in \HH$ and measurements $\{E^x_a\}_{a\in\A}$, $\{E^y_b\}_{b\in\B}$
for each $x\in \XX$ and $y\in \YY$ (for Alice and Bob respectively) such that
$E^x_a$ and $E^y_b$ commute for every $x,a,y,b$, and such that we have
\[P(x,y)(a,b) = \braket{\psi|E^x_aE^y_b|\psi}.\]

Suppose that such a set of measurements does exist.  Then we can consider a 
matrix $\Gamma$ whose rows and columns are indexed by formal products of our 
operators $E^x_a,E^y_b$, and whose entries are given by 
\[\Gamma_{S,T} = \braket{\psi|S^\dag T|\psi}.\]
This matrix is positive semidefinite, since for any 
vector $v=\sum_S v_S e_S$ (with basis vectors $e_S$) we have 
\begin{align*}
v^\dag \Gamma v &= \sum_{S,T} v_{S}^* \braket{\psi|S^\dag T|\psi} v_{T} \\
&= \braket{\psi|({\textstyle \sum_S} v_S S)^\dag ({\textstyle \sum_T} v_T T)|\psi} \\
&= \left\lVert {\textstyle \sum_T} v_T T \ket{\psi} \right\rVert^2 \\
&\geq 0.
\end{align*}

The matrix $\Gamma$ is infinite, so to formulate a finitary SDP problem we must
take a finite subset of its rows and columns: let the matrix $\Gamma_i$ consist
of those rows and columns of $\Gamma$ corresponding to formal products of at
most $i$ operators. We consider the problem whose constraints are that
$\Gamma_i\succcurlyeq 0$, together with additional constraints on the entries of
$\Gamma_i$ arising from the commutativity, orthogonality and idempotence properties of the
operators, and also from the desired values of
$P(x,y)(a,b)=\braket{\psi|E^x_aE^y_b|\psi}$. Clearly the $\Gamma_i$ arising from
a set of measurements realising the behaviour will be a feasible solution to
this problem, but the highly non-trivial main theorem
of~\cite{navascues2008convergent} (Theorem 8) is that the converse is also true:
if the problem is feasible for all $i$ then a suitable set of measurements
exists. Hence in particular if the behaviour is not quantum then we will find
that the problem is infeasible for some $i$.

Instead of specifying a fixed behaviour, we can also formulate some objective 
as a function of the $\braket{\psi|E^x_aE^y_b|\psi}$ (or rather the corresponding 
entries of $\Gamma_i$), and then the optimal values for increasing $i$ will give 
a sequence of tighter and tighter bounds, converging to the true optimum in the 
commuting operator model.  This is the main technique for bounding the entangled 
value of non-local games (see Section \ref{sec:nonlocal}), and this is what 
we will adapt below to obtain bounds on entangled min-entropy capacity.

\subsection{Bounds on min-entropy capacity}

As discussed above, to apply SDP techniques we will need to consider min-entropy 
capacity in the commuting operator model, which is stronger than the entangled 
model (Definition \ref{def:entcap}) but weaker than the non-signalling 
model (Definition \ref{def:genstrat}).  Whereas in the entangled model we 
specified that the Hilbert space 
on which Alice and Bob made their measurements could be separated into a part 
$\HH_\A$ held by Alice and a part $\HH_\B$ held by Bob, we will now drop this 
assumption and assume only that the measurements made by Alice commute with 
those made by Bob.  

\begin{definition}
Let $\C$ be an $n$-IC, and $K$ a random variable taking values on the set $\KK$.
A \emph{commuting operator joint strategy} for $\C$ is a pure state $\ket{\psi}$
in a (possibly infinite-dimensional) complex Hilbert space $\HH$, and sets
$\{A^{k,t}\}$ and $\{B^{t'}\}$ such that
\begin{enumerate}[(i)]
\item for every $k\in \KK$ and every $t\in (\A\times\XX)^{i}$ with $0\leq i < n$, 
$A^{k,t}=\{A^{k,t}_a\}_{a\in\A}$ is a measurement over $\HH_\A$,  
\item for every $t'\in (\B\times \YY)^i$ with $0\leq i < n$, 
$B^{t'}=\{B^{t'}_b\}_{b\in\B}$ is a measurement over $\HH_\B$, and
\item for every $k,t,t',a,b$ we have $A^{k,t}_aB^{t'}_b = B^{t'}_bA^{k,t}_a$.
\end{enumerate}
Denote the space of such strategies by $\Ss^\co_{\C,K}$.
\end{definition}

As usual we can define $\V$-leakage and $\V$-capacity $\LL_\V^\co(\C)$.  Note 
that any entangled strategy is trivially a commuting operator strategy; on the 
other hand a commuting operator strategy is still non-signalling and so we have
\begin{equation}\label{eq:commop}
\LL_\V^\ns(\C) \geq \LL_\V^\co(\C) \geq \LL_\V^*(\C) \geq \LL_\V(\C).
\end{equation}

The basic idea is that as before we consider a matrix $\Gamma$ with entries 
$\braket{\psi|S^\dag T|\psi}$, where $S$ and $T$ are formal products (of bounded 
length) of the operators defining our strategy.  The additional ingredient is 
that we are able to express the objective of min-entropy capacity as a linear 
function of the entries of $\Gamma$, or rather more precisely as a linear function of 
additional scalar variables which are subject to linear constraints.  This is 
done using the formula for min-entropy capacity given as Proposition 5.1 
of~\cite{braun2009quantitative}.

This formula states that if we have a (non-interactive) channel defined by a 
conditional probability matrix $p_{Y|X}$ then we have 
\begin{equation}\label{eq:mincap}\sup_{X} I_{-H_\infty}(X; Y) = \sum_{y\in\YY} 
\max_{x\in \XX} p_{Y|X}(y|x),\end{equation}
where the supremum is over probability distributions for $X$, with $Y$ obeying 
the conditional probabilities $p_{Y|X}(y|x)$.

Fixing strategies (i.e. sets of operators) for Alice and Bob fixes the
conditional distribution $\pi_B(T_{\phi_A(K),s_B})|K$, whose values we will see
can be expressed as linear functions of the entries of $\Gamma$. We then have
that $\LL_{-H_\infty}^\co(\C)$ is the supremum of
$I_{-H_\infty}(K;\pi_B(T_{\phi_A(K),s_B})$ over all choices of strategies and
all distributions for $K$, so in particular by (\ref{eq:mincap}) the capacity
corresponding to a given choice of strategies is given by
\[\sum_{t\in (\B\times\YY)^n} \max_k p_{\pi_B(T_{\phi_A(K),s_B})|K}(t|k).\]

Note that $\max$ is not a linear (or indeed convex) relation, and so this cannot 
be expressed directly in our SDP problem.  However, since $\KK$ and 
$(\B\times\YY)^n$ are finite sets, we can just exhaust over `guessing functions' 
$g:(\B\times\YY)^n\rightarrow \KK$, with the SDP for each maximising 
$\sum p(t|g(t))$.

Let $\C$ be an $n$-IC, $\KK$ a finite set and $g:(\B\times\YY)^n\rightarrow \KK$.  
Define the semidefinite programming problem $\mathcal{P}_i(\C,\KK,g)$ by the 
following variables:
\begin{itemize}
\item a matrix $\Gamma$, with entries $\Gamma_{S,T}$ for all strings 
$S,T$ in symbols $A_a^{k,t}$, $B_b^{t'}$, 0 and 1 of length at most $i$; $\Gamma_{S,T}$
represents $\braket{\psi|S^{\dag}T|\psi}$
\item variables $p_{t|k}$ for each trace $t\in (\B\times \YY)^n$ and each $k\in 
\KK$, representing the probability of observing trace $t$ conditional on the 
secret value $K=k$,
\end{itemize}
and objective
\[\text{maximise } \sum_{t\in(\B\times\YY)^n} p_{t|g(t)}.\]
The first constraints arises from the fact that all of the $p_{t|k}$ represent 
probabilities, and as discussed above the matrix $\Gamma$ is positive semidefinite.
\begin{itemize}
\item for all $t,k$, $0\leq p_{t|k} \leq 1$
\item $\Gamma \succcurlyeq 0$
\end{itemize}
The next constraints arise from the properties that $\Gamma$ should have if it 
is the matrix arising from some set of operators: for instance we will have 
$\braket{\psi|I|\psi}=1$ and $\braket{\psi|0|\psi}=0$.  More generally, if 
strings $S,T$ and $U,V$ are such that (interpreting the strings as products of 
operators) the orthogonality, idempotence and commutativity properties force 
$S^\dag T=U^\dag V$ then we should have $\Gamma_{S,T}=\Gamma_{U,V}$.
\begin{itemize}
\item $\Gamma_{1,1}=1$ and $\Gamma_{0,0}=0$
\item $\Gamma_{S,T}=\Gamma_{U,V}$ whenever $S^\dag T = U^\dag V$ under the following
relations: $(A_a^{k,t})^\dag=A_a^{k,t}=(A_a^{k,t})^2$ and 
$(B_b^{t'})^\dag=B_b^{t'}=(B_b^{t'})^2$; $A^{k,t}_aB_b^{t'} = B_b^{t'} A^{k,t}_a$;
$A^{k,t}_a A^{k,t}_{a'} = 0$ for all $a'\neq a$ and $B_b^{t'} B_{b'}^{t'} = 0$ 
for all $b'\neq b$; and $S1=S=1S$ and $S0=0=0S$ for all $S$
\end{itemize}
The final constraints express that the $p_{t|k}$ are indeed the conditional 
probabilities according to the formula in Section \ref{sec:entangleddef}.  Note 
that these are linear in the SDP variables, since the $f_i(t_i)(x_i,y_i)$ are 
fixed constants.
\begin{itemize}
\item for all $s\in (\B\times \YY)^n$ and $k\in \KK$, we have
\[p_{t|k} = \sum_{s:t=\pi_B(s)} \Gamma_{A_s^k,B_s} \prod_{i=1}^n f_i(s_i)(x_i,y_i),\]
with $A_t^k$ and $B_t$ the products as defined in Section \ref{sec:entangleddef}.
\end{itemize}

Write $\opt(\mathcal{P}_i(\C,\KK,g))$ for the optimal value of 
$\mathcal{P}_i(\C,\KK,g)$, and let $\opt_i(\C,\KK) = \max_g 
\opt(\mathcal{P}_i(\C,\KK,g)$.  Then we 
have that this converges to the commuting operator min-entropy capacity of $\C$.

\begin{theorem}
Let $\C$ be an $n$-IC.  Then
\[\LL_{-H_\infty}^\co(\C) = \lim_{i\rightarrow \infty} 
\log \opt_i(\C,(\B\times\YY)^n).\]
\end{theorem}
\begin{proof}
First observe that without loss of generality we may take $\KK=(\B\times \YY)^n$: 
indeed, if $|\KK|>|(\B\times \YY)^n|$ then there will elements $k\in \KK$ 
which are never Bob's optimal guess after the interaction, and assigning 
probability to these elements in the prior clearly cannot increase min-entropy 
leakage.  The upper bound on $\LL_{-H_\infty}^\co(\C)$ is then immediate, since 
a commuting operator strategy for $\C$ gives a feasible solution for 
$\mathcal{P}_i(\C,(\B\times \YY)^n,g)$ for all $i$ as described above (with 
$g(t)=\argmax_k p_{\pi_B(T_{\phi_A(K),s_B})|K}(t|k)$).

The lower bound is more delicate. First note that there are only finitely many
possible values for $g$ and so by passing to a subsequence we may assume that
$g$ is fixed. We then proceed essentially by the proof of Theorem 8
of~\cite{navascues2008convergent}. This shows that if we have a sequence of
feasible solutions (say with optimal values $v_i$), whose $\Gamma$ matrices we
denote $\Gamma^i$, then (viewing the $\Gamma^i$ as living in the space of
infinite matrices whose entries are indexed by \emph{all} strings $S$ and $T$,
extending $\Gamma^i$ with zeros as necessary), there is a pointwise convergent
subsequence whose limit is (say) the infinite matrix $\Gamma^\infty$, and
moreover there exists an (infinite-dimensional) Hilbert space $\HH$, state
$\ket{\psi}\in \HH$ and collection of operators $A^{k,t}_a$ and $B^{t'}_b$ such
that $\Gamma^\infty_{S,T} =\braket{\psi|S^{\dag} T|\psi}$ for all $S,T$.

Now all of the $p_{t|k}$ are probabilities and so are contained in the compact
set $[0,1]$. Hence passing to a further subsequence we may assume that all the
$p_{t|k}$ converge, and by continuity of the constraints we have that their
limit, say $p_{t|k}^\infty$ is a valid set of conditional probabilites for the
strategy corresponding to the operators obtained in the previous paragraph, with
advantage $\sum_t p_{t|g(t)}^\infty = \lim_{i\rightarrow \infty} v_i$, as
required.
\end{proof}

Note that for $j>i$ we have that any solution for $\opt(\mathcal{P}_j(\C,\KK,g))$
restricts to a solution for $\opt(\mathcal{P}_i(\C,\KK,g))$ with the same value 
of the objective, and so the $\log \opt_i(\C,(\B\times\YY)^n)$ are a descending 
sequence of upper bounds for $\LL_{-H_\infty}^\co(\C)$.

By (\ref{eq:commop}), the $\log \opt_i(\C,(\B\times\YY)^n)$ also give upper
bounds for $\LL_{-H_\infty}^*(\C)$. By Theorem \ref{thm:capnoncomp} there exist
channels $\C$ such that $\LL_{-H_\infty}^\co(\C) > \LL_{-H_\infty}^*(\C)$ (since
otherwise we could simultaneously enumerate upper bounds from above and
entangled strategies from below) and so the upper bounds do not converge to the
true value of $\LL_{-H_\infty}^*(\C)$, but since the question of whether such
channels exist is equivalent to the Connes Embedding Problem which was open for
more than 40 years, it seems reasonable to expect that this will not arise in
practice.

\section{Conclusions}\label{sec:conclusion}

\subsection{Interactive channels and non-local games}

In this work we have shown that there is a close connection between the 
communication capacity of interactive channels and the value of non-local games.  
In particular in Theorem \ref{thm:gameequiv} we have shown that for every game 
there exists a channel such that the entangled capacity of the channel 
corresponds to the entangled value of the game (and the same argument would give 
corresponding results for other classes such as commuting operator or 
non-signalling strategies).

What about going the other way? For the particular case of min-entropy capacity,
it does seem that one could transform a channel into a (multi-round) game, 
essentially by having Bob guess the value of the secret at the end (modulo the 
technical issue of the prior probability distribution over secrets not being 
specified \emph{a priori}); one could do the same for other $g$-leakage 
measures, using randomness outside the control of Alice and Bob to represent 
rewards between 0 and 1.  On the other hand it is difficult to see how to do 
this for general vulnerability measures, including in particular Shannon 
entropy---how can one express this as simple acceptance or rejection of a 
transcript?

It thus seems that interactive channel capacity is in some sense a 
generalisation of non-local games, where non-local games correspond specifically 
to capacity with respect to $g$-leakage measures.  Since non-local games have 
given rise to such a beautiful and useful theory, it seems reasonable to wonder 
whether a similarly rich theory may be available for other leakage measures.  A 
promising starting point would seem to be the Shannon entropy measure, which is 
on the one hand a natural measure but on the other not (so far as we can tell) 
encompassed by non-local games.

\subsection{Quantum QIF}

The formulation of quantitative measures of information flow was the beginning, 
not the end, of the subject of QIF.  Similarly, although this paper formulates a 
definition of entangled information flow and addresses some fundamental 
theoretical questions, it certainly does not claim to answer all the questions 
which are necessary to assess to what extent information leakage assisted by 
entanglement may constitute a threat in practice.  In particular, the systems we 
have considered have mainly been rather artificial, constructed from non-local 
games specifically to have the properties we want.  In the future it will be 
necessary to analyse more realistic systems to determine whether they may be 
affected by entanglement.  This is likely to require handling less abstracted 
models than that described in Section \ref{sec:classflow}, and finding pragmatic 
algorithms which are more efficient in practical cases than that described in 
Section \ref{sec:sdp}.

Finally, we do not by any means intend to suggest that approaches in the style 
of \cite{americo2020qqif}, in which the channels themselves may be quantum, are 
anything other than equally important as future directions for QIF in the quantum
realm.  Quantum networks may well become extremely relevant in the near or 
medium-term future, and indeed quantum key distribution systems already exist.  
We hope that in the future it may be possible to extend the approach of this 
paper to that setting, perhaps by extending Fig \ref{fig:flowfig} to allow 
quantum states as messages.

\bibliographystyle{IEEEtran}
\bibliography{mybib}
\vspace{12pt}

\end{document}